\documentclass[draftclsnofoot,12pt,onecolumn]{IEEEtran}       
\usepackage{amsmath, amsthm, amssymb}
\usepackage{epsfig}
\usepackage{latexsym}
\usepackage{graphicx}
\usepackage{cite}
\usepackage{epsfig}
\usepackage{subfigure}
\usepackage{color}

\newtheorem{theorem}{Theorem}
\newtheorem{lemma}{Lemma}

\newtheorem{corollary}{Corollary}  
\newtheorem{remark}{Remark}

\title{ Solutions to Integrals Involving the  Marcum $Q-$Function and Applications}

\author{
Paschalis C. Sofotasios, Sami~Muhaidat, George K. Karagiannidis\\ and Bayan S. Sharif 

\thanks{P. C. Sofotasios is with the Department of Electronics and Communications Engineering, Tampere University of Technology, 33101 Tampere, Finland and with the Department of Electrical and Computer Engineering, Aristotle University of Thessaloniki, 54124 Thessaloniki, Greece  \, (e-mail: p.sofotasios@ieee.org)  }

\thanks{S. Muhaidat is with the Department of Electrical and Computer Engineering, Khalifa University, PO Box 127788, Abu Dhabi, UAE and with
the Centre for Communication Systems Research, Department of Electronic Engineering, University of Surrey, GU2 7XH, Guildford, U.K. (e-mail:
muhaidat@ieee.org)}

\thanks{G. K. Karagiannidis is with the Department of Electrical and Computer Engineering, Khalifa University, PO Box 127788
Abu Dhabi, UAE and with the Department of Electrical and Computer Engineering, Aristotle University of Thessaloniki, 54124 Thessaloniki, Greece \, (e-mail: geokarag@ieee.org)}

\thanks{ B. S. Sharif,  is  with the Department of Electrical and Computer Engineering, Khalifa University, P.O. Box 127788, Abu
Dhabi, UAE. (e-mail:  bayan.sharif@kustar.ac.ae)} 
}

\begin{document}
\maketitle

\begin{abstract}
Novel analytic solutions are derived for  integrals that involve  the generalized Marcum $Q-$function, exponential functions and arbitrary powers.  Simple closed-form expressions are also derived for specific cases of the generic integrals. The offered expressions  are both convenient and versatile, which is particularly useful in applications relating to natural sciences and engineering,  including wireless communications and signal processing.  To this end, they  are employed in the derivation of the average probability of detection in energy detection of unknown signals over  multipath fading  channels as well as  of the channel capacity with fixed rate and channel inversion in  the case of correlated multipath fading and switched diversity. 
\end{abstract}


\section{Introduction}

The  generalized Marcum $Q{-}$function, $Q_{m}(a,b)$,  has been extensively involved in numerous areas of wireless communications including  digital communications over fading channels, information-theoretic analysis of multi-antenna systems, cognitive radio communications, radar systems,  \cite{J:Marcum_2, J:Nuttall_2, J:Simon, B:Alouini,  J:Brychkov2012, Bao, J:Himal, Ding} and references therein. Furthermore, its use  has enabled  the derivation of several tractable analytic expressions for important performance measures in communication theory 
\cite{B:Alouini}.  

The  derivation of tractable analytic expressions in natural sciences and engineering is typically a tedious, if not impossible, task because   cumbersome integrals   are  often encountered  \cite{J:Alouini, Math_7, J:Ran, Add, Ermolova, New_1, New_2, New_3, New_4, New_5, New_6, New_7, New_8, New_9}. This is also the case when the   Marcum $Q{-}$function is involved in integrands along with exponential functions and arbitrary power terms. Two such integrals  are the following:   

\begin{equation} \label{first}
 \mathcal{G}(k, m, a, b, p) =   \int_{0}^{\infty} x^{k-1} Q_{m}(a, b\sqrt{x}) e^{-px}{\rm d}x 
\end{equation}
and

\begin{equation} \label{Integral}
\mathcal{F}(k, m, a, b, p) =   \int_{0}^{\infty} x^{k-1} Q_{m}(a\sqrt{x}, b) e^{-px}{\rm d}x.  
\end{equation}
These  integrals have been widely employed in    the analysis of multi-channel  receivers with non-coherent and differentially coherent detection as well as  in the detection of unknown signals in   cognitive radio and radar systems  \cite{J:Alouini,  J:Janti, C:Beaulieu, J:Herath, Add_4, Maged, New_10, New_11, New_12, New_13, New_14, New_15, New_16, New_17, New_18, New_19, New_20} and the references therein. Based on this,  a recursive formula for \eqref{Integral}, that is restricted to only integer values of $k$ and $m$, was firstly reported in  \cite{J:Nuttall_2}. Likewise,   exact infinite series for \eqref{first} and \eqref{Integral}  were proposed in \cite{J:Ran} while a closed-form solution to \eqref{Integral} for integer values of $k$ was recently reported  in \cite{Add}.

Nevertheless, the existing expressions  for \eqref{first} and \eqref{Integral} are subject to validity restrictions, which limit the generality of the involved parameters and often render them inconvenient to use  in  applications of interest.  
  Motivated by this, the present work is devoted to the derivation of novel  closed-form expressions for  \eqref{first} and \eqref{Integral}, which are  more generic  and have a relatively tractable algebraic representation. These characteristics  render them   useful in several  analyses in natural sciences and engineering, including the broad areas  of   wireless communications and signal processing. To this end, they are subsequently employed in the derivation of  closed-form expressions for   the  following  indicative applications: 
$i)$ the average probability of detection  in energy detection  over Nakagami$-m$  multipath  fading channels which, unlike previous analyses, is valid for arbitrary values of $m$;  
 $ii)$ the channel capacity with channel inversion  and fixed rate in  arbitrarily correlated Nakagami$-m$ fading  conditions using switch-and-stay combining.  The derived expressions are validated  extensively through  comparisons with respective computer simulations results.

\section{Analytic Solutions to  Integrals Involving Power, Exponential and Marcum $Q-$functions}

\subsection{Closed-form Solutions to $\mathcal{G}(k, m, a, b, p) $} 

\begin{theorem}
For   $a, b \in \mathbb{R}$, $m \in \mathbb{N}$ and  $k, p \in \mathbb{R^{+}}$, the following closed-form representation  holds

\begin{equation} \label{G_1} 
 \mathcal{G}(k, m, a, b, p) =   \frac{\Gamma(k)}{p^{k}} -   \left( \frac{2}{b^{2} + 2p} \right)^{k} e^{-\frac{a^{2}}{2}} \Gamma(k) \qquad \qquad \qquad   \qquad \qquad \qquad   \qquad \qquad \qquad   \qquad \qquad \qquad   
  \end{equation}
  \begin{equation*}
\times \left[ \Phi_{1}\left(k, 1, 1; \frac{b^{2}}{b^{2} + 2p}, \frac{a^{2} b^{2}}{2b^{2} + 4p} \right) - \sum_{n = 0}^{m - 1} \frac{(k)_{n}}{n!} \left( \frac{b^{2}}{b^{2} + 2p} \right)^{n} \, _{1}F_{1}\left( k + n; n + 1; \frac{a^{2} b^{2}}{2b^{2} + 4p}\right) \right] 
  \end{equation*}

  \noindent
 where $\Gamma(\cdot)$,  $\Phi_{1}(\cdot)$ and $_{1}F_{1}(\cdot)$ denote the Euler Gamma function, the Humbert  hypergeometric function of the first kind and the  Kummer  hypergeometric function, respectively \cite{Yury_1}. 
\end{theorem}

\begin{proof}
The   series in \cite[eq. (10)]{J:Ran} can be  re-written  as follows\footnote{Equation (10) in  \cite{J:Ran} contains a typo since the $(a^{2} + 2p)$ term in the denominator should read as $(b^{2} + 2p)$. This  has been corrected in \eqref{Analysis_1}.}    

\begin{equation} \label{Analysis_1}
\begin{split}
\mathcal{G}(k, m, a, b, p) =&  \sum_{n=0}^{m-1} \frac{ \Gamma(k + n)   \, _{1}F_{1}\left( k + n; n + 1; \frac{a^{2}b^{2}}{2b^{2} + 4p} \right) }{n! b^{-2n} 2^{-k}   (b^{2} + 2p)^{k+n} e^{\frac{a^{2}}{2}} }    + \frac{\Gamma(k)}{p^{k}}  \\
& - \underbrace{ \frac{2^{k}}{e^{\frac{a^{2}}{2}}} \sum_{n=0}^{\infty} \frac{ \Gamma(k + n)  b^{2n}   \, _{1}F_{1}\left( k + n; n + 1; \frac{a^{2}b^{2}}{2b^{2} + 4p} \right) }{n!   (b^{2} + 2p)^{k+n}  }}_{\mathcal{A}(k, a, b, p)}. 
\end{split}
\end{equation}
Using  the infinite series in  \cite[eq. (9.14.1)]{B:Ryzhik}, it follows that

\begin{equation} \label{scalar_1}
\mathcal{A}(k, a, b, p) = \sum_{n = 0}^{\infty}\sum_{l = 0}^{\infty} \frac{ a^{2l} b^{2n + 2l}  \Gamma(n + k) e^{-\frac{a^{2}}{2}} (k + n)_{l} }{  n! l! 2^{l - k}  (b^{2} + 2p)^{n + k + l}   (1+n)_{l}}   
\end{equation}
where $(x)_{n}$ denotes the Pochhammer symbol. Given that  

\begin{equation}
(k + n)_{l} =   \frac{(k)_{n+l}}{(k)_{n}}
\end{equation}
 and 
 \begin{equation}
 (1 + n)_{l} =  \frac{(1)_{n+l}}{(1)_{n}}
 \end{equation}
  one obtains

\begin{equation} \label{Humbert_1}
\mathcal{A}(k, a, b, p) = \mathcal{B} \sum_{n = 0}^{\infty}\sum_{l = 0}^{\infty}  \frac{(k)_{n+l} (1)_{n}}{(1)_{n+l}} \frac{\left( \frac{b^{2}}{b^{2} + 2p} \right)^{n}}{n!} \frac{\left(\frac{a^{2} b^{2}}{2b^{2} + 4p} \right)^{l}}{l!} 
\end{equation}
where
\begin{equation}
\mathcal{B} = \frac{2^{k} \Gamma(k)  \exp(a^{2}/2)}{ (b^{2} + 2p)^{k}}
\end{equation}
Notably,  eq. \eqref{Humbert_1} can be expressed  in terms of the  Humbert function, $\Phi_1$,  namely

\begin{equation} \label{scalar_2} 
\mathcal{A}(k, a, b, p) =  \frac{2^{k} \, \Gamma(k) \, e^{-\frac{a^{2}}{2}} }{(b^{2} + 2p)^{k} \,  }  \Phi_{1}\left(k, 1, 1; \frac{b^{2}}{b^{2} + 2p}, \frac{a^{2} b^{2}}{2b^{2} + 4p} \right)
\end{equation}
Inserting \eqref{scalar_2} in \eqref{Analysis_1} yields \eqref{G_1}, which completes  the proof. 
\end{proof}
It is noted that  Humbert functions and their properties   have been  studied extensively  over the past decades  \cite{B:Ryzhik, Yury_1}.

\begin{theorem}
For $a, b \in \mathbb{R}$ and $m, p \in \mathbb{R}^{+}$ and $k \in \mathbb{N}$, the following closed-form expression holds

\begin{equation} \label{deuteri}  
\mathcal{G}(k, m, a, b, p) = \frac{\Gamma(k)}{p^{k}}  - \frac{\Gamma(k) \, b^{2m} e^{-\frac{a^{2}}{2}}   }{p^{k} (b^{2} + 2p)^{m}}     \sum_{l = 0}^{k-1} \frac{ (m)_{l} (2p)^{l}   }{l!    (b^{2} + 2p)^{l} }     \,_{1}F_{1}\left(l +m; m; \frac{a^{2} b^{2}}{2b^{2} + 4p} \right). 
\end{equation}
\end{theorem}

\begin{proof}
By integrating \eqref{first}   by parts, it follows that

\begin{equation}\label{T_2a}
\mathcal{G}(k, m, a, b, p)  =      \lim_{x \to \infty }  \, g(x)  \int \frac{x^{k} \, {\rm d}x}{x \, e^{px}}     -  \lim_{x \to 0}  \, g(x)  \int \frac{x^{k} \, {\rm d}x}{x \, e^{px}}    - \int_{0}^{\infty}  \left[ \int \frac{x^{k-1}}{ e^{px}} {\rm d}x \right]\, \frac{{\rm d}}{{\rm d}x}  Q_{m}(a, b \sqrt{x}) {\rm d}x   
\end{equation}
\noindent 
where 

\begin{equation}
g(x) = Q_{m}(a, b\sqrt{x})
\end{equation}

The   integrals in \eqref{T_2a} can be expressed in  terms of the  incomplete gamma function yielding
 
\begin{equation}\label{T_2b}
\mathcal{G}(k, m, a, b, p)  =    \lim_{x \to \infty} \,  g(x) \frac{\Gamma(k, p x)}{p^{k}}  -  \lim_{x \to 0}  \, g(x)  \frac{\Gamma(k, p x)}{p^{k}}  - \frac{1}{p^k} \int_{0}^{\infty}    \Gamma(k, p x) \, \frac{{\rm d}}{{\rm d}x} Q_{m}(a, b\sqrt{x}) {\rm d}x. 
\end{equation}
Recalling the identities 

\begin{equation}
\Gamma(a,\infty) = 0,
\end{equation}

\begin{equation}
\Gamma(a,0) = \Gamma(a)
\end{equation}
and

\begin{equation}
Q_{m}(a, 0) = 1
\end{equation}
  as well as setting $u = \sqrt{x}$ in \eqref{T_2b} along with utilizing  \cite[eq. (9)]{J:Nuttall_2} and \cite[eq. (8.352.4)]{B:Ryzhik}, it follows that
 
\begin{equation}\label{T_2c}
\mathcal{G}(k, m, a, b, p)  =    \frac{\Gamma(k)}{p^{k}} \times  \left[ 1  -    a^{1-m} e^{-\frac{a^{2}}{2}} \sum_{l=0}^{k-1} \frac{   p^l}{l!  b^{2l}} \int_{0}^{\infty} x^{m-2l}    e^{  - \left( \frac{p}{b^{2}} + \frac{1}{2} \right) x^{2}}     I_{m-1}(ax)     \, {\rm d}x   \right].  
\end{equation} 
The  above integral can be expressed in closed-form using \cite[eq. (6.621.1)]{B:Ryzhik}.  To this effect, eq.  \eqref{T_2c} is deduced, which completes the proof. 
\end{proof}

\subsection{Closed-form Solutions to $\mathcal{F}(k, m, a, b, p) $}

A closed-form expression for \eqref{Integral} was reported  in  \cite{Add} for the case that $k$ is integer and $m$ is   arbitrary real,  namely

\begin{equation} \label{crowncom}
\begin{split}
\mathcal{F}(k, m, a, b, p) =& \frac{\Gamma(k)  \Gamma\left(m, \frac{b^{2}}{2} \right)  }{p^{k} \Gamma(m)}   \\
& +  \frac{a^{2} b^{2m} \Gamma(k)  e^{-\frac{b^2}{2}} }{m! p^{k} 2^{m}(a^{2} + 2p)} \sum_{l = 0}^{k - 1} \left( \frac{2p}{a^{2} + 2p} \right)^{l}   \, _{1}F_{1}\left(l + 1; m + 1; \frac{a^{2} b^{2}}{2 a^{2} + 4p} \right).  
\end{split}
\end{equation}
\noindent 
Likewise, in what follows we derive a closed-form expression  for the useful   case that $m$ is  integer and $k$ is   arbitrary real.

\begin{theorem}
For  $k, p \in \mathbb{R^{+}}$, $m \in \mathbb{N}$ and $a, b \in \mathbb{R}$, the following closed-form expression is valid
 
\begin{equation} \label{T_1}
\mathcal{F}(k, m, a, b, p) =  \frac{\Gamma(k)}{p^{k}}  - \frac{2^{k} \Gamma(k)  e^{- \frac{b^{2}}{2}} \Phi_{2}\left( 1, k, 1; \frac{b^2}{2}, \frac{a^{2} b^{2}}{2 a^{2} + 4p}\right) }{(a^{2} + 2p)^{k}}  \qquad \qquad \qquad 
  \end{equation}
  \begin{equation*}
\qquad \quad  + \frac{\Gamma(k) 2^{k}  e^{-\frac{b^2}{2}}}{(a^{2} + 2p)^{k}} \sum_{n = 0}^{m-1} \frac{  b^{2n} }{n! 2^{n} }    \, _{1}F_{1}\left(k; n + 1; \frac{a^{2} b^{2}}{2a^{2} + 4p} \right)  
  \end{equation*}
  where   $\Phi_{2}(\cdot)$ is  the Humbert    function of the second kind  \cite{B:Ryzhik}. 
\end{theorem}

\begin{proof}
The  series in \cite[eq. (9)]{J:Ran} can be also expressed as
 
\begin{equation} \label{derivation}
\begin{split}
\mathcal{F}(k, m, a, b, p) =& \frac{\Gamma(k)}{p^{k}} + \sum_{n = 0}^{m-1} \frac{ \Gamma(k) _{1}F_{1}\left(k; n + 1; \frac{a^{2}b^{2}}{2a^{2} + 4p} \right)}{n! b^{-2n} 2^{n-k} (a^{2} + 2p)^{k}  e^{\frac{b^{2}}{2}}  }  \\
&  -  \underbrace{ \frac{\Gamma(k) 2^{k} e^{-\frac{b^2}{2}}}{(a^{2} + 2p)^{k}} \sum_{n = 0}^{\infty} \frac{ b^{2n}}{n!  2^{n}  }  \,   _{1}F_{1}\left(k; n + 1; \frac{a^{2}b^{2}}{2a^{2} + 4p} \right) }_{\mathcal{H}(k, a, b, p)}. 
\end{split}
\end{equation}
\noindent 
Applying  \cite[eq. (9.14.1)]{B:Ryzhik} in $\mathcal{H}(k, a, b, p)$  and given that   

\begin{equation}
(n+1)_{l} = \frac{\Gamma(n + l + 1)}{\Gamma(n+1)} = \frac{(1)_{n+l}}{(1)_{n}}
\end{equation}
one obtains
 
\begin{equation} \label{derivation_2}
\mathcal{H}(k, a, b, p) =  \frac{e^{-\frac{b^{2}}{2}}2^{k} \Gamma(k) }{(a^{2} + 2p)^{k}} \sum_{n=0}^{\infty}\sum_{l=0}^{\infty} \frac{ (1)_{n} (k)_{l} }{(1)_{n+l} } \frac{ \frac{b^{2n}}{2^{n}}}{n!} \frac{ \left( \frac{a^{2} b^{2}}{2a^{2} + 4p} \right)^{l}}{l!}.  
\end{equation}
Notably, the above expression can be expressed in terms of the Humbert hypergeometric function of the second kind yielding 
 
\begin{equation} \label{derivation_3}
\mathcal{H}(k, a, b, p) =  \frac{e^{-\frac{b^{2}}{2}}2^{k} \Gamma(k) }{(a^{2} + 2p)^{k}} \Phi_{2}\left( 1, k, 1; \frac{b^2}{2}, \frac{a^{2} b^{2}}{2 a^{2} + 4p}\right).   
\end{equation}
Inserting  \eqref{derivation_3} into  \eqref{derivation} yields \eqref{T_1},  concluding the proof. 
\end{proof}
\noindent

\subsection{Specific Cases of  $\mathcal{F}(k, m, a, b, p) $ and $\mathcal{G}(k, m, a, b, p) $}

Simple  expressions are  derived for  $k = 1$, $a = 0$ and $b=0$. 

\subsubsection{The case that $k=1$} In this special case   it follows that 

\begin{equation} \label{second_special_1}
\mathcal{G}(1, m, a, b, p) =   \int_{0}^{\infty}   Q_{m}(a, b\sqrt{x}) e^{-px} {\rm d}x
\end{equation} 
and
 
\begin{equation} \label{Integral_special_1}
\mathcal{F}(1, m, a, b, p) =   \int_{0}^{\infty}  Q_{m}(a\sqrt{x}, b) e^{-px}{\rm d}x.  
\end{equation}
Equation  \eqref{second_special_1} is given by \cite[eq. (16)]{J:Ran}. In the same context, 
a   generic  closed-form expression for \eqref{Integral_special_1} is derived below. 
\begin{lemma}
For $ a, m \in \mathbb{R}$ and  $  b, p \in \mathbb{R}^{+} $, the following closed-form expression is  valid\footnote{The integral in \cite[eq. (1)]{Ermolova} was recently evaluated in closed-form. However, this solution does not account for $\mathcal{F}(k, m, a, b, p)$ since the two integrals would be equal only when $I_{0}(0) = 1$. Yet,  this can be achieved  for $\mu_{2} = 1$ and $c=0$, which     eliminates  the power term and as a consequence, the integral reduces  to \eqref{Integral_special_1}, which is simply expressed in closed-form  in \eqref{special_1}.} 
 
\begin{equation} \label{special_1}
\mathcal{F}(1, m, a, b, p) =   \frac{\Gamma\left(m, \frac{b^{2}}{2} \right)}{p \Gamma(m)} +  \frac{ a^{2} e^{-\frac{pb^{2}}{a^{2} + 2p}}  \gamma\left(m, \frac{a^{2} b^{2}}{2 a^{2} + 4p} \right) }{p a^{2m }  \Gamma(m)  (a^{2} + 2p)^{1 - m}} 
\end{equation}
where $\gamma(a, x)$ denotes  the lower incomplete   gamma function. 
\end{lemma}
\begin{proof}
By setting $k=1$ in \eqref{crowncom}, it immediately follows that
 
\begin{equation} \label{specific_1}
\mathcal{F}(1, m, a, b, p) = \frac{ \Gamma\left(m, \frac{b^{2}}{2} \right)}{p  \Gamma(m)} 
 +  \frac{a^{2} b^{2m}   \, _{1}F_{1}\left( 1; m + 1; \frac{a^{2} b^{2}}{2 a^{2} + 4p} \right) }{m! p 2^{m  } (a^{2} + 2p) e^{\frac{b^{2}}{2}}} 
\end{equation}
Notably, using  the  following hypergeometric function identity
 
\begin{equation}
\,_{1}F_{1}(1; n; x) = x^{1 - n} (n - 1) e^{x} \gamma(n - 1, x) 
\end{equation}
and recalling that 

\begin{equation}
\gamma(a, x) = \Gamma(a) - \Gamma(a, x)
\end{equation}
and

\begin{equation}
\frac{m!}{m} = (m-1)! = \Gamma(m)
\end{equation}
  equation \eqref{specific_1} can be equivalently expressed as
 
\begin{equation} \label{specific_2}
\mathcal{F}(1, m, a, b, p) =   \frac{\Gamma\left(m, \frac{b^{2}}{2} \right)}{p \, \Gamma(m)} + \frac{ (a^{2} + 2p)^{m-1} e^{ - \frac{a^{2} b^{2}}{a^{2} + 2p}}}{ p \,  a^{2m - 2}  }  - \frac{   (a^{2} + 2p)^{m-1}e^{-\frac{a^{2} b^{2}}{a^{2} + 2p}}   }{ p \, a^{2m-2} \Gamma(m)   } \Gamma\left( m, \frac{a^{2} b^{2}}{2 a^{2} + 4p} \right).  
\end{equation}
To this effect and performing long but basic algebraic manipulations \eqref{specific_2} reduces to \eqref{special_1}, which completes the proof. 
\end{proof}

\begin{remark}
A similar expression can be  obtained through \eqref{T_1}.  
\end{remark}

\subsubsection{The case that $a=0$} In this special case it follows that
 
\begin{equation} \label{special_2}
\mathcal{G}(k, m, 0, b, p) =   \int_{0}^{\infty} x^{k-1} Q_{m}(0, b\sqrt{x}) e^{-px}{\rm d}x  
\end{equation}
and
 
\begin{equation} \label{special_5a}
\mathcal{F}(k, m, 0, b, p)  = Q_{m}(0, b) \int_{0}^{\infty} x^{k - 1}e^{-px} {\rm d}x. 
\end{equation}

\begin{lemma}
For $k, p \in \mathbb{R}^{+}$, $m \in \mathbb{N}$ and $b \in \mathbb{R}$,  the following closed-form representations hold
 
\begin{equation} \label{special_3}
\mathcal{G}(k, m, 0, b, p) = \frac{2^k}{(b^{2} + 2p)^{k}}\sum_{l = 0}^{m - 1} \frac{b^{2l}  \Gamma(k+l)}{l! (b^{2} + 2p)^{l}}  
\end{equation}
and
 
\begin{equation}\label{special_5}
\mathcal{F}(k, m, 0, b, p)  = \frac{\Gamma(k) }{p^{k}} e^{-\frac{b^{2}}{2}}  \sum_{l = 0}^{m-1} \frac{b^{2l}  }{l! 2^{l}  }.  
\end{equation}
\end{lemma}

\begin{proof}

Applying \cite[eq. (2)]{J:Nuttall_2} in \eqref{special_2}  one obtains
 
\begin{equation} \label{special_4}
\mathcal{G}(k, m, 0, b, p) = \sum_{l=0}^{m-1} \frac{b^{2l}}{l! 2^{l}} \int_{0}^{\infty} x^{k + l - 1} e^{-x\left( p + \frac{b^{2}}{2} \right)} {\rm d}x. 
\end{equation}
Evidently, the above integral can be expressed in closed-form in terms of the Euler gamma function.  This is also the case for $\mathcal{F}(k, m, 0, b, p) $ since $Q_{m}(0, b)$  is not a part of the integrand. As a result, by performing a necessary change of variables in \cite[eq. (8.310.1)]{B:Ryzhik} and substituting  \eqref{special_4} and \eqref{special_2} one obtains \eqref{special_3} and \eqref{special_5a}, respectively, which completes the proof.  
\end{proof}

\subsubsection{The case that $b=0$}  In this special case it follows that 
 
\begin{equation} \label{b=0_1}
\mathcal{G}(k, m, a, 0, p)    = Q_{m}(a, 0)\int_{0}^{\infty} x^{k - 1} e^{-px} {\rm d}x
\end{equation}
and
 
\begin{equation} \label{b=0_2}
  \mathcal{F}(k, m, a, 0, p)  = \int_{0}^{\infty} x^{k - 1} Q_{m}(a \sqrt{x},0) e^{-px} {\rm d}x. 
\end{equation}

\begin{lemma}
For $p \in \mathbb{R}^{+}$ and $a, m, k \in \mathbb{R}$, the following simple closed-form expression is valid 
 
\begin{equation} \label{special_c}
\mathcal{G}(k, m, a, 0, p)  =  \mathcal{F}(k, m, a, 0, p)   = \frac{\Gamma(k)}{p^{k}} . 
\end{equation}
\end{lemma}

\begin{proof}
The proof follows with the aid of  the identity  $Q_{m}(a, 0) \triangleq 1$,  in \cite[eq. (1)]{J:Nuttall_2} as well as  \cite[eq. (8.310.1)]{B:Ryzhik}. 
\end{proof}

 \section{Applications in Wireless Communications}
 
As already mentioned, the derived expressions for \eqref{first} and \eqref{Integral} can be used in applications relating to  natural sciences and engineering, including wireless communications and signal processing. Based on this, they are employed  in the derivation of simple expressions for the detection of unknown signals in cognitive radio and radar systems as well as  for the  channel capacity using  switched diversity.  To this end,  a  closed-form expression is firstly derived for the average probability of detection over  Nakagami$-m$ fading channels. Likewise,  a closed-form expression  is derived for the channel  capacity with channel inversion and fixed rate in switch-and-stay combining (SSC) under  correlated Nakagami$-m$ fading conditions.

   \begin{figure}[ht]
\centering
   \includegraphics[width =12cm, height =9cm] {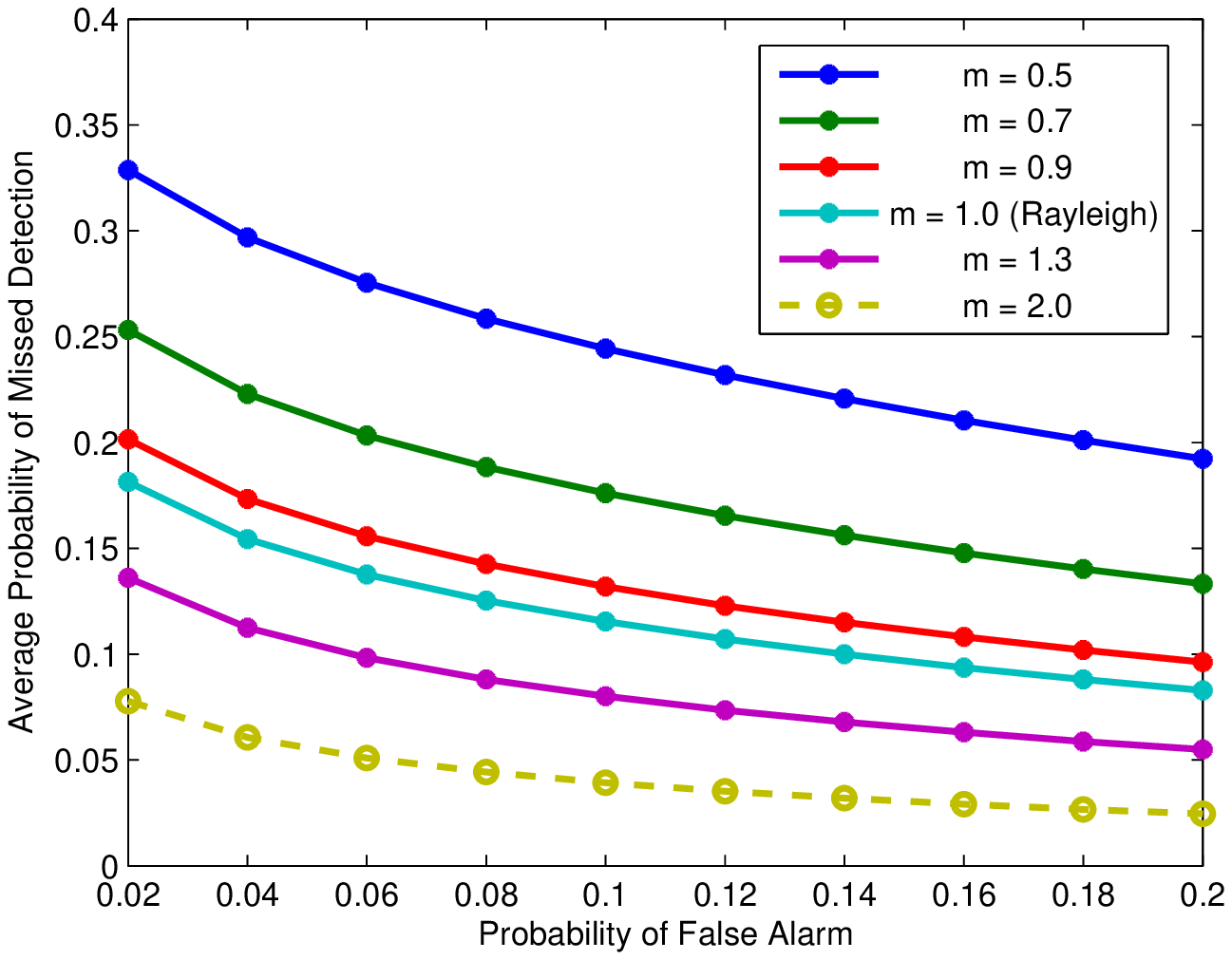}
 \label{myfigure1}
\caption{ ROC curve for $u = 5$, $\overline{\gamma} = 15$dB and different values of $m$. }
\end{figure}

 \subsection{Energy detection over Nakagami$-m$ fading channels with arbitrary  values of $m$}

The  detection  of unknown signals is  modeled as a binary hypothesis-testing problem, where $H_0$ and $H_1$ denote the cases that a  signal  is absent or present, respectively. The corresponding test statistic is typically represented by the central chi-square and the non-central chi-square distributions, respectively,  and is compared with an energy threshold, $\lambda$ \cite{J:Alouini}.  
\newpage 
 \begin{corollary}
 For $\overline{\gamma}, \lambda \in \mathbb{R}^{+}$, and either $m \geq 0.5$ and $u \in \mathbb{N}$, or $m \in \mathbb{N}$ and $u \in \mathbb{R}^{+}$, the average probability of detection over Nakagami$-m$ fading channels can be expressed as  
 
 \begin{equation} \label{app}
 \overline{P}_d = \frac{m^m}{\overline{\gamma}^{m} \Gamma(m)} \mathcal{F}\left(m, u, \sqrt{2}, \sqrt{\lambda}, \frac{m}{\overline{\gamma}}\right). 
 \end{equation}
 \end{corollary}
 
 \begin{proof}
The  probability of false alarm and probability of detection  in  additive white Gaussian noise  are given by  $P_{f} = \Gamma(u, \lambda {/}2)$ and $P_{d} = Q_{u}(\sqrt{2\gamma}, \sqrt{\lambda})$, respectively, where $u$ and $\gamma$ denote the time-bandwidth product and the instantaneous signal-to-noise ratio (SNR), respectively   \cite{Add, Ermolova, J:Janti}. It is recalled that in  energy detection  over fading channels, the $P_d$ is averaged over the fading statistics.  To this effect,  for the case Nakagami$-m$ fading channels in  \cite[eq. (2. 21)]{B:Alouini} is represented as follows
 
\begin{equation} \label{app_1}
\overline{P}_d =    \frac{m^{m} }{\overline{\gamma}^{m}  \Gamma(m)} \int_{0}^{\infty }  \gamma^{m-1} Q_{u}(\sqrt{2\gamma}, \sqrt{\lambda}) e^{-\frac{m \gamma}{\overline{\gamma}}}   {\rm d}\gamma.  
\end{equation}
Evidently, the above integral can be expressed in terms of \eqref{crowncom} or \eqref{T_1}.  This yields  \eqref{app},  which completes the proof.  
 \end{proof}
 
Notably, the offered expression can account for arbitrary values of $m$, contrary to existing analyses that assume integer values of $m$ for simplicity. Fig. $1$ illustrates the corresponding probability of missed detection versus probability of false alarm ROC curve for different values of $m$. One can notice the sensitivity of $m$, particularly for small values, and thus, the usefulness of the offered expression also in practical scenarios.

\subsection{Capacity with channel inversion and fixed rate over correlated Nakagami$-m$ fading  using switch-and-stay combining}

Channel capacity under different transmission policies is  particularly useful in achieving  certain quality of service requirements. In this context, channel inversion with fixed rate (CIFR) has been rather useful as it ensures a constant SNR at the receiver through  adaptation of the transmit power.  This method relies on  fixed-rate modulation and fixed code design, which renders its implementation  relatively simple \cite{B:Alouini}.

\begin{figure}[ht]
\centering
   \includegraphics[width =12cm, height =9cm] {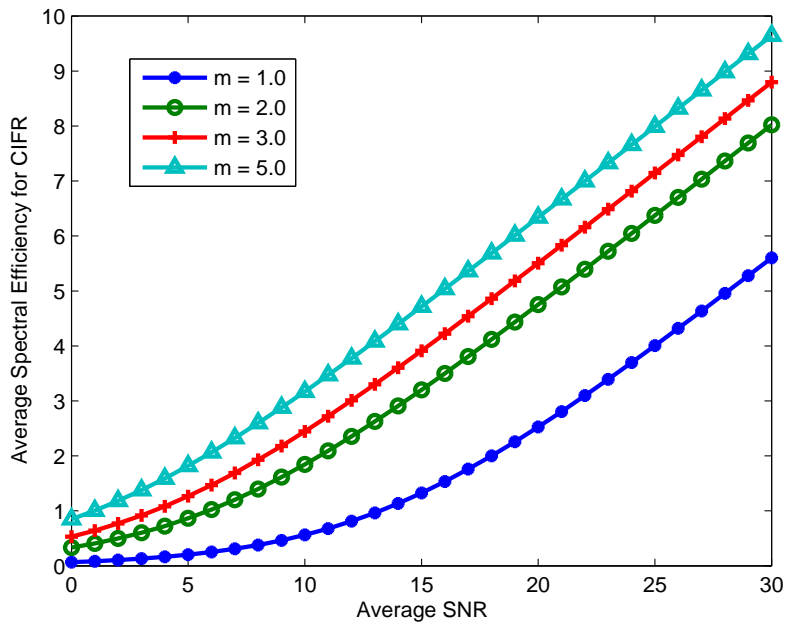}
 \label{myfigure1}
\caption{ Average SE vs $\overline{\gamma}$ for CIFR under correlated Nakagami$-m$ fading with SSC for $\gamma_T = 0$dB, $\rho = 0.5$ and different values of $m$.}
\end{figure}

\begin{corollary}
For    $\overline{\gamma}, B \in \mathbb{R}^{+}$,  $m \geq \frac{1}{2}$ and $0 \leq \rho < 1$, the  capacity  with channel inversion and  fixed rate over  correlated Nakagami$-m$ fading channels with SSC can be expressed as
 
\begin{equation} \label{C_CIFR}
C_{\rm CIFR} = B \log_{2}\left(1 + \frac{1}{\mathcal{R}(m, \overline{\gamma}, \gamma_{T}, \rho)} \right) 
\end{equation}
where 
 
\begin{equation}
\mathcal{R}(m, \overline{\gamma}, \gamma_{T}, \rho) =   \frac{m}{\overline{\gamma} \Gamma(m)} \Gamma\left( m - 1, \frac{m \gamma_{T}}{\overline{\gamma}} \right)   -   \frac{m^m}{\overline{\gamma}^{m} \Gamma(m)} \mathcal{F}\left(m-1, m,  \sqrt{\frac{2 m \rho}{(1 - \rho) \overline{\gamma}}}, \sqrt{\frac{2 m \gamma_{T}}{(1 - \rho) \overline{\gamma}}}, \frac{m}{\overline{\gamma}} \right)   
\end{equation}
with  $\rho$ and $\gamma_T$ denoting the  correlation coefficient and the predetermined SNR switching threshold, respectively. 
\end{corollary}

\begin{proof}
The CIFR over fading channels is defined as  \cite{J:Goldsmith}  
 
\begin{equation}
C_{\rm CIFR} = B\log_{2}\left(1 + \frac{1}{\int_{0}^{\infty} \frac{p_{\gamma}(\gamma)}{\gamma} {\rm d}\gamma} \right). 
\end{equation}

\noindent 
In the case of switched diversity and correlated Nakagami$-m$ fading, the  PDF of the SSC output is given by \cite[eq. (9.334)]{B:Alouini}. 
By also setting

\begin{equation}
 \mathcal{R} = \int_{0}^{\infty} \frac{p_{\gamma_{\rm SSC}}(\gamma)}{\gamma} \,  {\rm d}\gamma
 \end{equation}
   it follows that  
   
   \begin{equation}
\mathcal{R}= \int_{0}^{\infty}A(\gamma){\rm d}\gamma + \int_{\gamma_T}^{\infty}p_{\gamma}(\gamma) {\rm d}\gamma
 \end{equation}
     where $A(\gamma)$ is given in \cite[eq. (9.335)]{B:Alouini}.  To this effect and  using \cite[eq. (2. 21)]{B:Alouini} one obtains 
 
\begin{equation} \label{SSC_2}
\begin{split}
\mathcal{R}(m, \overline{\gamma}, \gamma_{T}, \rho)  =& \frac{m^m}{\overline{\gamma}^{m} \Gamma(m) } \int_{\gamma_T}^{\infty} \gamma^{m-2}e^{-\frac{m \gamma}{\overline{\gamma}}} \, {\rm d}\gamma   \\
&  - m^{m}\int_{0}^{\infty} \frac{\gamma^{m-2} \, e^{-\frac{m \gamma}{\overline{\gamma}}} }{\overline{\gamma}^{m}\Gamma(m)  } Q_{m}\left( \sqrt{\frac{2 m \rho \gamma}{(1 - \rho)\overline{\gamma}}},  \sqrt{\frac{2 m  \gamma_T}{(1 - \rho)\overline{\gamma}}}\right)  {\rm d}\gamma. 
\end{split}
\end{equation}
The first two integrals in \eqref{SSC_2} can be expressed in terms of the gamma functions, whereas the third integral has the algebraic form of \eqref{Integral}. Therefore, by performing the necessary change of variables yields \eqref{C_CIFR}, which completes the proof. 
\end{proof}
The behavior of the corresponding average spectral efficiency versus average SNR is illustrated in Fig. 2 for  different values of $m$  with fixed values of $\gamma_T$ and $\rho$. The significant effect of the severity of fading on   $C_{\rm CIFR}$ is clearly observed.

\section{Conclusion}
\label{conc}
Novel closed-form expressions were derived for two Marcum $Q-$function integrals which are both simple and generic. Simple analytic expressions for   involved special cases were also derived in closed-form. These expressions are tractable and are expected to be useful in analyses relating to natural sciences and engineering, including wireless  communications and signal processing. To this end, they were employed in the analysis of energy detection in RADAR and cognitive radio systems as well as in  the channel   capacity with channel inversion and fixed rate over correlated  multipath fading channels.


\bibliographystyle{IEEEtran}

\end{document}